\documentclass[lettersize,journal]{IEEEtran}
\usepackage{amsmath,amsfonts}
\usepackage{algorithmic}
\usepackage{algorithm}
\usepackage{array}
\usepackage[caption=false,font=normalsize,labelfont=sf,textfont=sf]{subfig}
\usepackage{textcomp}
\usepackage{stfloats}
\usepackage{url}
\usepackage{verbatim}
\usepackage{graphicx}
\usepackage{cite}
\usepackage{amsthm}
\begin{document}

\newtheorem{theorem}{Theorem}
\newtheorem{lemma}{Lemma}
\newtheorem{definition}{Definition}
\newtheorem{corollary}{Corollary}
\newtheorem{proposition}{Proposition}

\title{Doppler Invariant CNN for Signal Classification}


\author{
\begin{tabular}[t]{c@{\extracolsep{8em}}c} 
Avi~Bagchi  &  Dwight~Hutchenson \\
\textit{University of Pennsylvania} & \textit{MIT Lincoln Laboratory} \\ 
Philadelphia, PA, USA & Lexington, MA, USA \\
\end{tabular}

\thanks{DISTRIBUTION STATEMENT A. Approved for public release. Distribution is unlimited.
This material is based upon work supported by the Department of the Air Force under Air Force Contract No. FA8702-15-D-0001 or FA8702-25-D-B002. Any opinions, findings, conclusions or recommendations expressed in this material are those of the author(s) and do not necessarily reflect the views of the Department of the Air Force.
© 2025 Massachusetts Institute of Technology.
Delivered to the U.S. Government with Unlimited Rights, as defined in DFARS Part 252.227-7013 or 7014 (Feb 2014). Notwithstanding any copyright notice, U.S. Government rights in this work are defined by DFARS 252.227-7013 or DFARS 252.227-7014 as detailed above. Use of this work other than as specifically authorized by the U.S. Government may violate any copyrights that exist in this work.
}
}

\markboth{}%
{Shell \MakeLowercase{\textit{et al.}}: A Sample Article Using IEEEtran.cls for IEEE Journals}

\IEEEpubid{}


\maketitle

\begin{abstract}
Radio spectrum monitoring in contested environments motivates the need for reliable automatic signal classification technology. Prior work highlights deep learning as a promising approach, but existing models depend on brute-force Doppler augmentation to achieve real-world generalization, which undermines both training efficiency and interpretability. In this paper, we propose a convolutional neural network (CNN) architecture with complex-valued layers that exploits convolutional shift equivariance in the frequency domain. To establish provable frequency bin shift invariance, we use adaptive polyphase sampling (APS) as pooling layers followed by a global average pooling layer at the end of the network. Using a synthetic dataset of common interference signals, experimental results demonstrate that unlike a vanilla CNN, our model maintains consistent classification accuracy with and without random Doppler shifts despite being trained on no Doppler-shifted examples. Overall, our method establishes an invariance-driven framework for signal classification that offers provable robustness against real-world effects. 
\end{abstract}


\section{Introduction}
\IEEEPARstart{R}{adio} spectrum monitoring is a crucial aspect of defense operations in wireless communication systems \cite{rml22}. In crowded and contested environments, however, there is often a lack of knowledge of the types of signals present, severely limiting the ability to counter the interference effects \cite{surveyamc}. Automatic signal classification technology addresses this challenge, attempting to classify unknown signals from the received samples. In the pre-machine learning era, research proposed extracting a small set of features to distinguish commonly seen signal types \cite{overtheair}. Such techniques were shown to be effective in modulation recognition where higher-order moments reveal differences in constellation shapes \cite{4thorder} . Approaches centered on feature selection, however, do not scale to real-world environments that contain a much larger set of signal types (e.g., interference signals) and potential channel disturbances \cite{rml22,revolution}.   

Deep learning algorithms provide a promising alternative for signal classification as they  implicitly learn a large amount of features from raw, complex baseband data and thus, distinguish between far more signal types even at low signal-to-noise ratio (SNR) \cite{rml22,overtheair}. The RADIOML synthetic dataset generator \cite{radiomldata} has driven academic and industrial research in this domain, leading to recurrent neural network, generative adversarial neural network, and auto-encoder implementations to classify signals \cite{rml22}. One successful model is the convolutional neural network (CNN), which achieves high classification accuracy even at low SNRs and under a variety of realistic channel conditions \cite{overtheair}. 

With any deep learning approach, however, there are two major challenges: generalization and interpretability. To achieve the former, prior work trains a model on a massive dataset \cite{revolution} that captures the wide range of real-world disturbances \cite{rml22}. Yet even after surmounting this computational challenge, the generalization gap for real-world applications remains unclear. Interpretability is challenging because it is unclear how many training examples with channel-induced perturbations are required for the model to learn robustness to such transformations. 

In this work, we harness the learning capability of CNNs for signal classification while also bridging the generalization and interpretability gaps. In particular, we construct and exploit an invariance---the property that the model output remains the same even after a transformation to its input \cite{cohen2016group}. We target the transformation induced by a Doppler shift, as such shifts often confuse signal classification models in real-world applications \cite{fusion}. Our contribution is as follows:

\begin{enumerate}
    \item We adapt the RADIOML22 CNN architecture \cite{rml22} to construct an interference signal classifier with provable Doppler-shift invariance, thereby improving both generalizability and interpretability. 
    \item Prior work \cite{rml22, overtheair,fusion} applies Doppler shifts to an arbitrary number of training samples, inflating the dataset size. Our training dataset remains small as it contains no Doppler-shifted examples.
\end{enumerate}

\section{System Model}

The complex baseband representation of a radio frequency (RF) signal can be expressed in terms of in-phase and quadrature components, representing the cosine and sine projections respectively. Concretely, the baseband signal is $x(t)=I(t)+jQ(t)$ while the passband signal is modulated onto a higher frequency at $x_p(t)=\Re (x(t)e^{j2\pi f_c t})$. Taking discrete-time samples of $x(t)$ yields $x[n]=I[n]+jQ[n]$, providing us with IQ samples \cite{rml22}. 

Let $\Omega$ be the set of signal types. Following the data dimensions of RADIOML22  \cite{rml22}, we generate IQ samples for 2000 frames with each frame containing 128 samples at various SNR values.    

$$\{(\omega, \rho) : \omega \in \Omega , \rho \in \{-20,-18\ldots 18,20\}\}$$

$$(\omega, \rho) \mapsto X_{\omega,\rho} \in \mathbb{R}^{2000 \times 2 \times 128}$$

To simulate a contested environment, we create a synthetic dataset with common interference signals: $\Omega=\{\mathrm{tone},\mathrm{hopping\_tone},\mathrm{chirp},\mathrm{noise},\mathrm{bpsk},\mathrm{qpsk},\mathrm{8psk}\}$ (Table \ref{tab:signal_definitions}). All signals are generated in baseband with the receiver tuned to $f_c=0$. Furthermore, the signals are modeled as being received by a receiver that was not synchronized in time, frequency, or phase with the transmitter.  

\begin{table}[t]
\centering
\caption{Signal Type Definitions}
\label{tab:signal_definitions}
\begin{tabular}{p{1.9cm} p{5.9cm}}
\hline
\textbf{Signal Type} & \textbf{Definition} \\
\hline

Tone &
Constant amplitude, continuous wave (CW) at a constant frequency within Nyquist sampling range. \\

Hopping Tone &
Constant–amplitude CW with stepped frequencies. Step length is constant within each frame but varies across frames. All of the frequency steps across the dataset are within the Nyquist sampling range. \\

Chirp &
Frequency that varies linearly with time, between a minimum and maximum frequency which both lie within the Nyquist sampling range. Once the maximum frequency is reached, the signal wraps back to the minimum frequency. The rate of frequency change is constant within a single frame, but varies across all frames in the dataset. \\

Noise &
Partial-band Gaussian noise. Created by passing white Gaussian noise through a bandpass filter with a passband entirely within the Nyquist sampling range. Passband frequencies are constant within a single frame, but vary across all frames. \\

M-PSK &
M-ary phase shift keying modulated signal with M=2,4, or 8 (i.e. BPSK, QPSK, or 8PSK). Symbol durations and modulation order are constant within a single frame, but vary across all frames in the dataset. The symbol durations are chosen such that the bandwidth of the signal does not extend beyond the Nyquist sampling range. \\
\hline
\end{tabular}
\end{table}

The goal is to train a CNN $f_{\theta}: \mathbb{R}^{2 \times 128} \to \Omega$. Of particular interest in this study is how transformations to the input affect the output of $f_{\theta}$. Following \cite{cohen2016group}, for transformations $T$ and $T'$, $f_{\theta}$ is \emph{equivariant} to $T$ if $$f_{\theta}(T(x))=T'(f_{\theta}(x))$$ where $T$ and $T'$ need not be distinct. Furthermore, $f_{\theta}$ is \emph{invariant} to a transformation $T$  if $$f_{\theta}(T(x))=f_{\theta}(x) $$.

The goal of this paper is to adapt the CNN architecture to construct invariance to a Doppler transformation, which is a change in frequency. A Doppler shift in the time domain is represented as:
$$x_d(t) = x(t)\, e^{j 2 \pi f_d t}$$

where $f_d$ is the frequency shift in baseband.

Let the DFT have $N$ frequency bins with spacing $\Delta f = f_s/N$, where $f_s$ is the sampling rate. A Doppler shift $f_d$ that is an integer multiple of $\Delta f$ produces a circular shift of $m$ bins,
\[
m = \frac{f_d}{\Delta f}.
\]
The DFT of the Doppler-shifted signal is therefore
\[
X_d[k] = X[(k - m) \bmod N].
\]

\section{Proposed Technique}

For an input in the time domain, convolutional layers are inherently equivariant to shifts in time \cite{lecun, cohen2016group}. To construct equivariance in frequency shifts, we must first map the input into the frequency domain with a Fast Fourier Transform (FFT). To handle feeding complex valued frequencies into the network, we adapt the convolutional layer and activation function to handle complex values \cite{survey} (Table \ref{tab:model_architecture}). Let $x[n] = x_r[n] + jx_i[n]$ be a complex-valued input and define $w_r$ and $w_i$ as real-valued weight matrices.

Following \cite{complexconv}, the complex 1-D convolution is defined as 
\[
\begin{aligned}
\Re\{\mathrm{ComplexConv}(x[n])\}
    &= \mathrm{Conv1D}(x_r, w_r)
       - \mathrm{Conv1D}(x_i, w_i), \\
\Im\{\mathrm{ComplexConv}(x[n])\}
    &= \mathrm{Conv1D}(x_r, w_i)
       + \mathrm{Conv1D}(x_i, w_r).
\end{aligned}
\]
 
Following \cite{deepcomplex,complexrelu}, the complex activation is applied elementwise as 
\[
\mathrm{ComplexReLU}(x[n])
    = \mathrm{ReLU}(x_r[n]) 
      + j\,\mathrm{ReLU}(x_i[n]).
\]
 
Following \cite{cvldsp}, the complex linear layer is defined as 
\[
\begin{aligned}
\Re\{\mathrm{ComplexLinear}(x[n])\}
    &= \mathrm{Linear}(x_r[n], w_r)
       - \mathrm{Linear}(x_i[n], w_i), \\[0.2em]
\Im\{\mathrm{ComplexLinear}(x[n])\}
    &= \mathrm{Linear}(x_r[n], w_i)
       + \mathrm{Linear}(x_i[n], w_r).
\end{aligned}
\]

Pooling layers follow each convolutional layer in a CNN and execute downsampling \cite{alexnet}. Recent research has shown that these layers can disrupt the aforementioned equivariance as small shifts can yield significantly different downsampled values and thus, final classification \cite{chaman, zhang}. In Algorithm~\ref{alg:complex_apspool}, we follow the approach of \cite{chaman} by constructing a pooling layer---adaptive polyphase sampling (APS)---that preserves shift equivariance (see Fig.~\ref{fig:redblue} for intuition).

\begin{algorithm}[H]
\caption{ComplexAPS1d}
\label{alg:complex_apspool}
\begin{algorithmic}[1]
\REQUIRE Input tensor $x \in \mathbb{C}^{B \times C \times L}$, stride $s$
\ENSURE Output tensor $y \in \mathbb{C}^{B \times C \times \lceil L/s \rceil}$
\STATE $\mathit{out\_arr} \gets [\,]$
\STATE $\mathit{out\_norm\_arr} \gets [\,]$
\FOR{$i = 0$ to $s-1$}
    \STATE $x^{(i)} \gets x[:, :, i::s]$ \COMMENT{Polyphase component starting at index $i$}
    \STATE Append $x^{(i)}$ to $\mathit{out\_arr}$
    \STATE $n^{(i)} \gets \|x^{(i)}\|_2$ \COMMENT{$\ell_2$ norm over all elements}
    \STATE Append $n^{(i)}$ to $\mathit{out\_norm\_arr}$
\ENDFOR
\STATE $max\_index \gets \arg\max(\mathit{out\_norm\_arr})$
\RETURN $\mathit{out\_arr}[max\_index]$
\end{algorithmic}
\end{algorithm}

\begin{figure}[h!]
    \centering
    \includegraphics[width=0.5\textwidth]{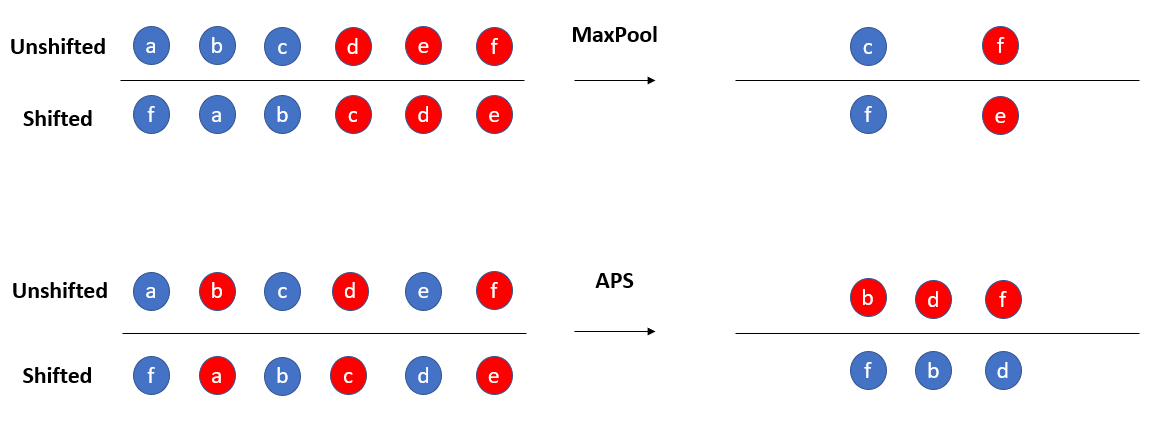}
    \caption{MaxPool (kernel=stride=3) vs APS (stride=2) for frequencies $a<b<c<d<e<f$. The MaxPool operation yields two distinct sets before and after a one bin shift while Algorithm \ref{alg:complex_apspool} yields the same set regardless of the shift.}
    \label{fig:redblue}
\end{figure}

Harnessing the equivariance of the earlier layers, we apply a global average pool in our last layer, ensuring that the final classification will be the same regardless of any shift (Table \ref{tab:model_architecture}). We can now assert that our architecture is invariant to a pure Doppler shift (please refer to the Appendix \ref{app:proof} for a formal proof). 

\begin{table*}[t]
\centering
\caption{Doppler Invariant CNN Architecture}
\label{tab:model_architecture}

\resizebox{\textwidth}{!}{
\begin{tabular}{|l|c|c|c|c|c|c|c|c|c|c|}
    \hline 
    \textbf{Description} & \textbf{Input} & 
    \textbf{CC1} & \textbf{APS1} & \textbf{CC2} & \textbf{APS2} & 
    \textbf{CC3} & \textbf{APS3} & \textbf{GP/Flatten} & 
    \textbf{FC/Softmax} \\
    \hline
    \textbf{Output Dim.} & 
    $1\times p^*$ &
    $64\times p^*$ &
    $64\times \left\lceil \frac{p^*}{s} \right\rceil$ &
    $32\times\left\lceil \frac{p^*}{s} \right\rceil$ &
    $32\times \left\lceil \frac{p^*}{s^2} \right\rceil$ &
    $16\times  \left\lceil \frac{p^*}{s^2} \right\rceil$ &
    $16\times  \left\lceil \frac{p^*}{s^3} \right\rceil$ &
    16 & 7 \\
    \hline 
    \textbf{Model} & \multicolumn{9}{p{18cm}|}{
        \textbf{Before entering the network:}
        \begin{itemize}
            \item Padding of size $p$ is applied to both sides of each IQ sample (length 128), giving $p^* = 128 + 2p$.
            \item FFT is applied to all IQ samples; the full complex values are fed into the network.
        \end{itemize}

        \textbf{Complex Convolution (CC) layers:}
        \begin{itemize}
            \item Implemented with two real Conv1D layers.
            \item Conv1D parameters: left pad = 1, right pad = 2, kernel size = 4, stride = 1.
            \item Includes ReLU activation.
        \end{itemize}

        \textbf{Adaptive Polyphase Sampling (APS)} layers use stride $s$.

        \textbf{Global Average Pooling (GP)} is followed by flattening.

        \textbf{Fully Connected (FC)} layer produces a 7-class softmax output.
     } \\
    \hline 
    \textbf{Training} & \multicolumn{9}{p{18cm}|}{
        Optimizer: Adam (learning rate = 0.001). \newline
        Loss: CrossEntropyLoss. \newline
        Batch size: 256. \newline
        Epochs: 15.
    } \\
    \hline 
\end{tabular}
}
\end{table*}

\begin{theorem}
Let $N_i$ denote the signal length entering APS layer $i$. If
$N_i \bmod s = 0$ for all APS layers, then $f_\theta$ is invariant with
respect to a pure Doppler shift, i.e.,
\[
f_\theta(T(x[n])) = f_\theta(x[n]),
\]
where $T$ denotes a Doppler transformation corresponding to an integer
bin shift $m \in \mathbb{Z}$.
\label{thm:1}
\end{theorem}

\section{Numerical Results}

We begin by evaluating Theorem \ref{thm:1} experimentally. That is, after training our new architecture solely on clean IQ samples, we divide the test set into two parts: clean data and data with frequency bin shifts applied to simulate a pure Doppler effect. In line with the analytical proof, Table \ref{tab:doppler_delta} shows that classification accuracy does not change after the pure Doppler shift. 

\begin{table}[h]
\centering
\caption{Doppler Invariant CNN (APS stride = 2): Signal Classification Accuracies Before and After Pure Doppler Shift (Shift 20 bins)}
\begin{tabular}{|l|c|c|c|}
\hline
\textbf{Signal Type} & \textbf{Before} & \textbf{After} & $|\Delta|$ \\
\hline
chirp         & 0.94 & 0.94 & 0 \\
hopping\_tone & 0.93 & 0.93 & 0 \\
tone          & 0.99 & 0.99 & 0 \\
noise         & 0.97 & 0.97 & 0 \\
bpsk          & 0.71 & 0.71 & 0 \\
qpsk          & 0.23 & 0.23 & 0 \\
8psk          & 0.20 & 0.20 & 0 \\
\hline
\textbf{Total} & -- & -- & {\textbf{0}} \\
\hline
\end{tabular}
\label{tab:doppler_delta}
\end{table}

In reality, however, a Doppler shift does not correspond to an exact bin shift. Recalling Theorem \ref{thm:1}, if $m \not \in \mathbb{Z}$ (i.e. the resolution in the frequency domain is too coarse to align with the Doppler-induced shift), there is spectral leakage, causing energy to be split between two bins. Clearly, our architecture cannot account for this since APS sampling is equivariant to a bin shift, not the values themselves in the bins changing. To remedy this problem, we add zero padding of length $p$ to the input on both sides in the time domain, yielding $p^*=128+2p>128$ frequency bins after taking the FFT. This modification enables the model to be invariant to finer grained Doppler shifts. 

Padding, however, changes the length of the frequency array $N$, and from Theorem \ref{thm:1}, we know it is advantageous to have $N \bmod s=0$. We want this condition to hold through every APS layer, motivating the \textit{padding condition}: we choose $p$ such that 

   \begin{itemize}
    \item $p^* \bmod s =0$
    \item $\left\lceil \frac{p^*}{s} \right\rceil \bmod s =0$
    \item $\left\lceil \frac{p^*}{s^2} \right\rceil \bmod s =0.$
\end{itemize}

\begin{table}[t]
\caption{Absolute Accuracy Change Totaled Across Signal Types}
\label{tab:abs_acc_change}
\centering
\begin{tabular}{|l|r|r|r|r|}
\hline
\textbf{Padding} & \textbf{Stride=2} & \textbf{Stride=3} & \textbf{Stride=4} & \textbf{Stride=5} \\
\hline
0   & 0.501 & 2.490 & 0.913 & 1.803 \\
10  & 0.180 & 0.903 & 0.631 & 1.573 \\
20  & 0.150 & 0.698 & 0.295 & 1.689 \\
30  & 1.119 & 0.956 & 0.277 & 0.658 \\
40  & \textbf{0.086} & 2.163 & 0.137 & 2.660 \\
50  & 0.985 & 1.876 & 2.000 & 2.918 \\
60  & 0.223 & 2.101 & 0.545 & 1.175 \\
70  & 1.320 & 0.797 & 0.236 & 0.644 \\
80  & \textbf{0.077} & 0.329 & 0.339 & 0.108 \\
90  & 1.252 & 1.814 & 1.968 & 1.298 \\
100 & 0.146 & 1.680 & 1.904 & 1.874 \\
110 & 1.240 & 1.301 & 0.648 & 0.212 \\
120 & \textbf{0.060} & 1.476 & 0.099 & 0.473 \\
130 & 1.185 & 1.540 & 2.392 & 0.263 \\
140 & 0.112 & 1.553 & 1.700 & 2.324 \\
150 & 1.134 & 1.579 & 0.798 & 2.569 \\
160 & 0.093 & 1.114 & 0.142 & 0.271 \\
170 & 1.175 & 1.374 & 1.919 & 0.251 \\
180 & \textbf{0.090} & 1.476 & 1.774 & 0.192 \\
190 & 1.116 & 1.163 & 0.134 & 2.641 \\
200 & \textbf{0.052} & 1.156 & 1.204 & 2.619 \\
210 & 1.047 & 1.483 & 1.415 & 1.734 \\
220 & 0.101 & 2.037 & 1.195 & 0.145 \\
230 & 1.136 & 1.215 & 1.349 & 0.365 \\
240 & 0.090 & 1.268 & 0.574 & 1.459 \\
250 & 1.115 & 1.234 & 1.533 & 2.224 \\
260 & \textbf{0.077} & \textbf{0.084} & 1.666 & 0.218 \\
270 & 1.052 & 1.577 & 1.168 & 0.624 \\
280 & \textbf{0.080} & 1.518 & \textbf{0.048} & 1.145 \\
290 & 1.103 & 1.293 & 1.549 & 1.770 \\
300 & \textbf{0.089} & 1.255 & 1.046 & 1.628 \\
\hline
\end{tabular}
\end{table}

We perform a parameter sweep over many possible padding and stride values, measuring the total absolute change in accuracy across all signal types (we show plots for each signal type in Appendix \ref{app:graphs}). In Table \ref{tab:abs_acc_change}, the ten lowest total accuracy changes are in bold. Notably, all entries in bold have padding values that follow the property introduced by the padding condition. Furthermore, most of the lowest changes in accuracy are from a stride of 2 and there are none from a stride of 5 because a higher value of $s$ is less likely to satisfy the aforementioned padding condition. The combination with the lowest accuracy change is $p_{opt}=280$ and $s_{opt}=4$. 

We use the vanilla CNN from the RADIOML22 \cite{rml22} paper as baseline. The model is similar to our architecture (Table \ref{tab:model_architecture}) in that it has three convolutional layers with ReLU activation, but it instead uses maximum pooling rather than APS sampling. The model also uses dropout and batch normalization. In Table \ref{tab:cnn_comparison}, we observe that after applying a random Doppler between 1 Hz and 5000 Hz to the testing set, there is a considerable drop in accuracy for the chirp, tone, bpsk, and qpsk signal types when using the vanilla CNN. With the Doppler invariant model, although there is a drop in base accuracy for the M-PSK signal types (likely due to training in the frequency domain), the accuracies are robust to the Doppler shifts as the changes in accuracy are negligible.

\begin{table}[t]
\centering
\caption{Comparison of Signal Classification Accuracies Before and After Random Doppler Shift (1 Hz to 5000 Hz)}
\label{tab:cnn_comparison}

\resizebox{\columnwidth}{!}{
\begin{tabular}{|l|c|c|c||c|c|c|}
\hline
{\textbf{Signal Type}} 
& \multicolumn{3}{c||}{\textbf{Vanilla}} 
& \multicolumn{3}{c|}{\textbf{Doppler Invariant (Stride=4, Pad=280)}} \\
\cline{2-7}
& \textbf{Before} & \textbf{After} & $|\Delta|$ 
& \textbf{Before} & \textbf{After} & $|\Delta|$ \\
\hline
chirp         & 0.95 & 0.04 & 0.91 & 0.893 & 0.912 & 0.019 \\
hopping\_tone & 0.90 & 0.90 & 0.00 & 0.925 & 0.927 & 0.002 \\
tone          & 0.96 & 0.60 & 0.36 & 0.983 & 0.983 & 0.000 \\
noise         & 0.96 & 0.97 & 0.01 & 0.956 & 0.958 & 0.002 \\
bpsk          & 0.89 & 0.05 & 0.84 & 0.577 & 0.585 & 0.008 \\
qpsk          & 0.38 & 0.24 & 0.14 & 0.082 & 0.088 & 0.006 \\
8psk          & 0.43 & 0.50 & 0.07 & 0.265 & 0.254 & 0.011 \\
\hline
\textbf{Total} & -- & -- & \textbf{2.33} & -- & -- & \textbf{0.048} \\
\hline
\end{tabular}
}
\end{table}

Despite these promising results, there are tradeoffs of the proposed technique. First, as Figure \ref{fig:runtime} demonstrates, increasing the padding increases the length of the signal and thus the time to compute a convolution. By contrast, Figure \ref{fig:runtime} also shows that increasing the stride decreases the amount of data downsampled, thus reducing the number of samples moving on to the next layer and the training time.  A very high stride early in the network, however, can contribute to information loss and inhibit the model from learning certain features. Thus, a $p_{opt}$ and $s_{opt}$ that are optimal from the perspective of invariance may not be optimal in certain  applications.

\begin{figure}[h!]
    \centering
    \includegraphics[width=0.5\textwidth]{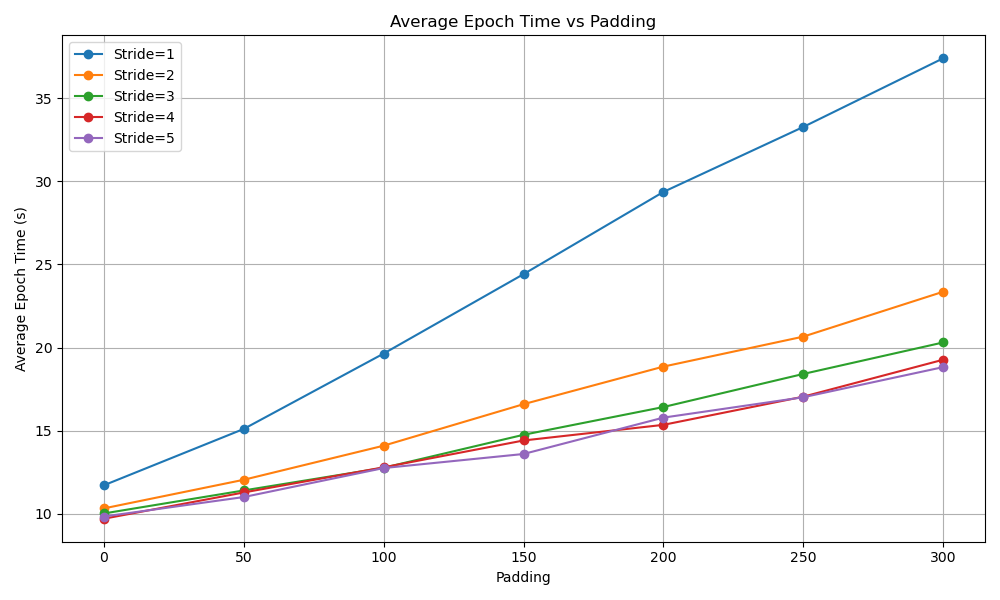}
    \caption{Average Epoch Time (s) vs Padding}
    \label{fig:runtime}
\end{figure}

\section{Conclusion}

This paper introduced a Doppler invariant CNN for signal classification. Training in the frequency domain and the use of APS created a provably frequency bin shift invariant model. When faced with fractional frequency bin shifts, our analysis showed that an appropriately chosen padding and stride value still allowed for near-perfect invariance. 

The main limitation of this work is that while there is no change in accuracy following a Doppler shift, training in the frequency domain lowers the classification accuracy overall, especially for the M-PSK modulations. Future work should construct a frequency-domain architecture that matches the performance in the time domain. 

Similar to prior studies using deep learning for signal classification, it is essential that the proposed model is tested in real-world environments and not just with synthetic data. We hypothesize that constructing invariant architectures will generalize better in such environments as opposed to data augmentation approaches. Future work should construct architectures invariant to other real-world effects such as time shifts, jamming signals, and fading. Whether one can weave these invariance properties together in the same architecture such that one invariance does not undermine another motivates future research.

\appendices
\section{Proof of Theorem \ref{thm:1}}\label{app:proof}
\begin{lemma}
Let $N$ denote the length of the frequency array. If $N \bmod s = 0$, then
\[
\mathrm{APS}(T(X[k])) = \mathrm{APS}(X[k]),
\]
where $T$ denotes a circular bin shift by an integer $m$.
\label{lem:aps}
\end{lemma}

\begin{proof}
Let $P = \{0,\ldots,s-1\}$ index the polyphase components, and let
$f:X[k]\to P$ map each array element to its polyphase component. Since
$N \bmod s = 0$, the polyphase pattern repeats exactly around the
circle, so
\[
f(X[(i+1)\bmod N]) = (f(X[i])+1)\bmod s
\]
for all $i$. Thus, a circular right shift by one rotates the assignment
of indices to offsets: the set of elements in each polyphase component
is unchanged, but now appears at offset $(j+1)\bmod s$ instead of $j$.
Algorithm~\ref{alg:complex_apspool} computes the $\ell_2$ norm $n^{(j)}$
of each polyphase component and selects the index
$\operatorname*{arg\,max}_j n^{(j)}$. A circular shift merely permutes
the list $\{n^{(0)},\ldots,n^{(s-1)}\}$, so the same polyphase component
attains the maximum norm before and after the shift (though possibly at
a different offset index). Repeating this argument for a shift by $m$ bins via
induction proves the lemma.
\end{proof}

\begin{proof}[Proof of Theorem \ref{thm:1}]
By Lemma~\ref{lem:aps}, under the condition $N_i \bmod s = 0$, each APS layer selects
the same polyphase component before and after the shift. We already know that convolutional layers are shift equivariant. Thus, the
sets entering the global average pooling layer for $T(x[n])$ and $x[n]$ are identical. The global average pooling layer is
invariant under permutations and produces the same value for $T(x[n])$ and $x[n]$. Therefore,  $f_\theta(T(x)) = f_\theta(x)$, establishing invariance
to pure Doppler shifts.

\end{proof}

\section{Additional Graphs}
\label{app:graphs}
\begin{figure}[h!]
    \centering
    \includegraphics[width=0.5\textwidth]{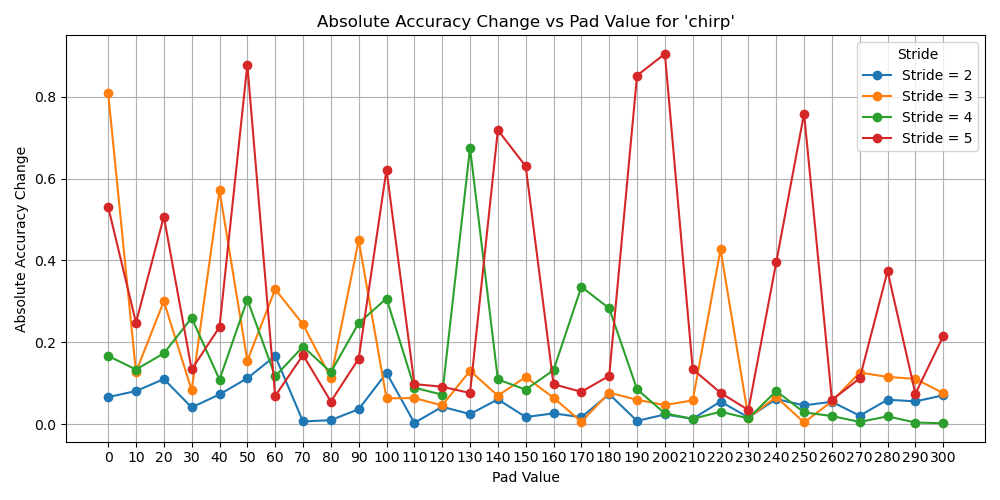}
    \caption{Absolute Accuracy Change (Chirp) vs Padding}
    
\end{figure}

\begin{figure}[h!]
    \centering
    \includegraphics[width=0.5\textwidth]{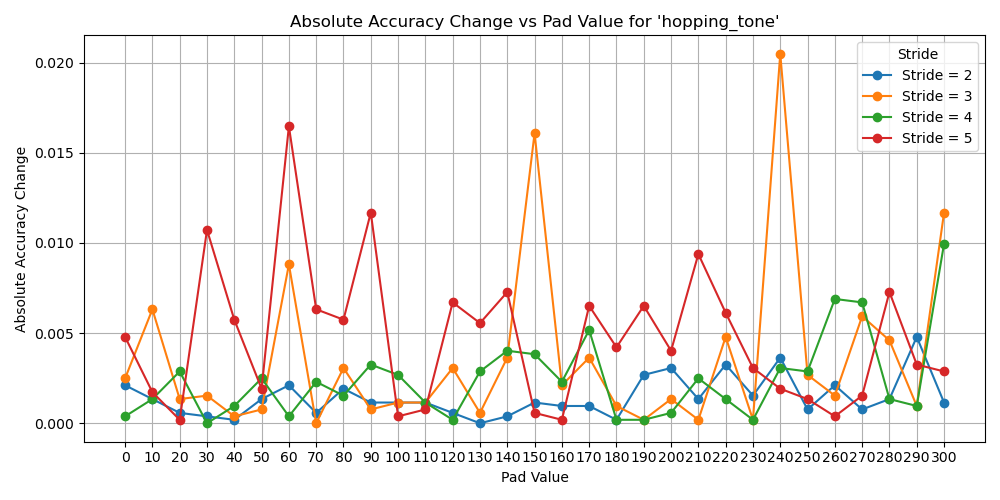}
    \caption{Absolute Accuracy Change (Hopping Tone) vs Padding}
    
\end{figure}

\begin{figure}[h!]
    \centering
    \includegraphics[width=0.5\textwidth]{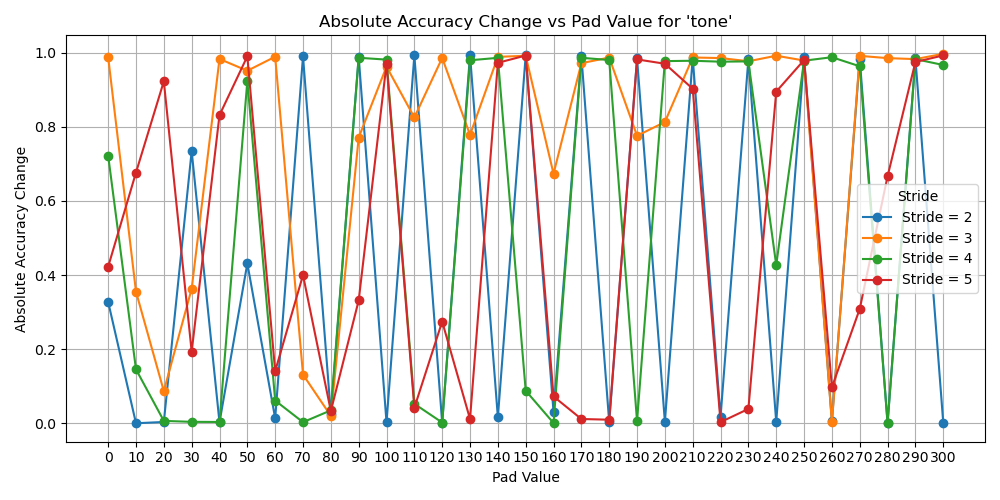}
    \caption{Absolute Accuracy Change (Tone) vs Padding}
    
\end{figure}

\begin{figure}[h!]
    \centering
    \includegraphics[width=0.5\textwidth]{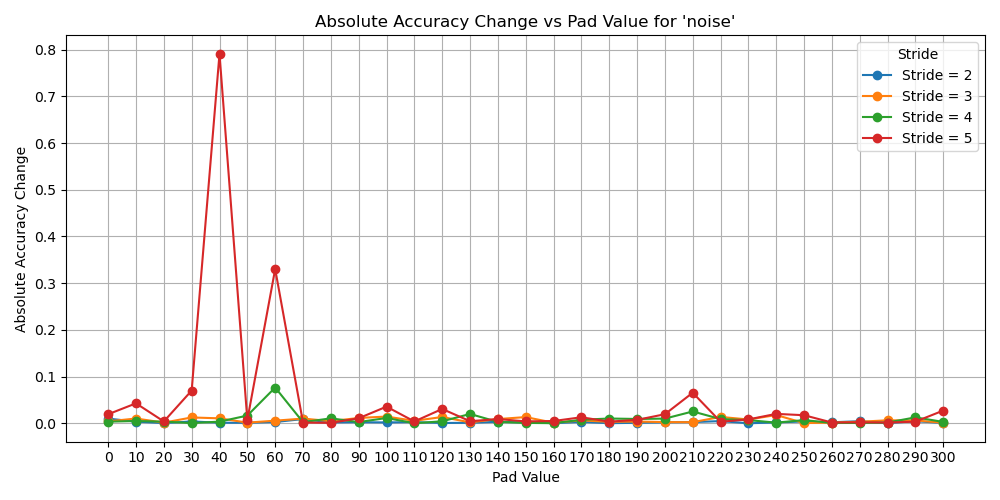}
    \caption{Absolute Accuracy Change (Noise) vs Padding}
    
\end{figure}

\begin{figure}[h!]
    \centering
    \includegraphics[width=0.5\textwidth]{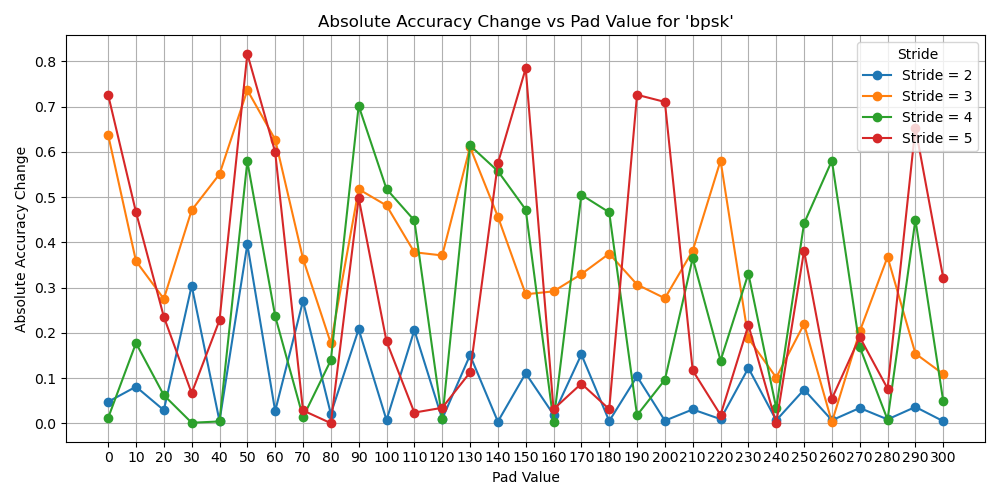}
    \caption{Absolute Accuracy Change (BPSK) vs Padding}
    
\end{figure}

\begin{figure}[h!]
    \centering
    \includegraphics[width=0.5\textwidth]{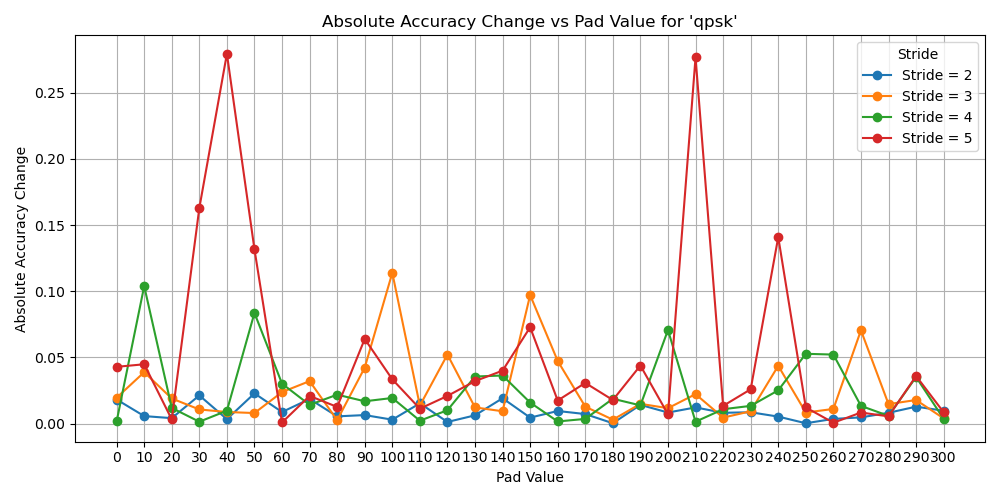}
    \caption{Absolute Accuracy Change (QPSK) vs Padding}
    
\end{figure}

\begin{figure}[h!]
    \centering
    \includegraphics[width=0.5\textwidth]{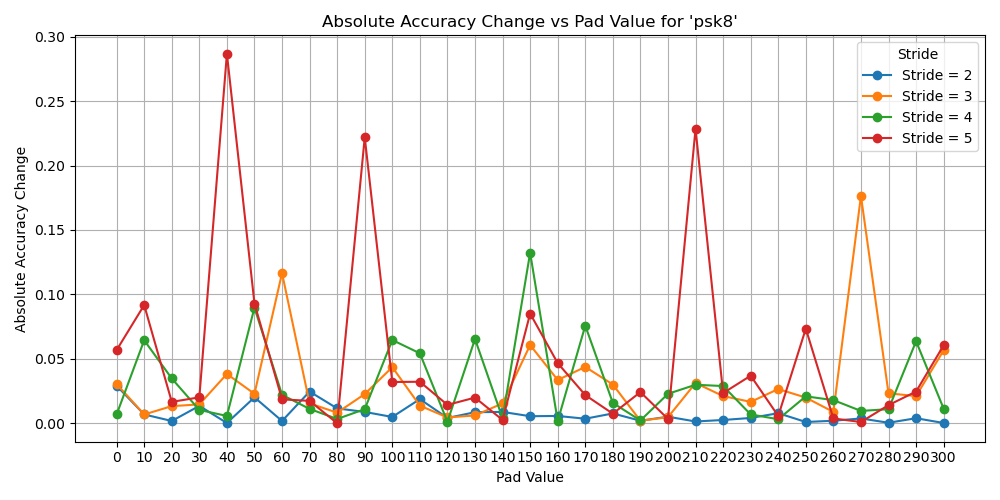}
    \caption{Absolute Accuracy Change (8PSK) vs Padding}
    
\end{figure}

\newpage

 




\vfill

\end{document}